\documentclass[a4paper, 12pt]{article}
\usepackage{a4wide} 
\usepackage{graphicx}
\usepackage{pgfplots}
\usepackage{here, amsmath, latexsym, amssymb, bm, ascmac, mathtools, multicol, tcolorbox, subfig, amsthm, natbib}

\mathtoolsset{showonlyrefs}
\theoremstyle{definition}

\newtheorem{definition}{Definition}
\newtheorem{theorem}{Theorem}
\newtheorem{proposition}{Proposition}

\newtheorem{lemma}{Lemma}

\usepackage[colorlinks=true, linkcolor=blue, citecolor=blue]{hyperref}
\usepackage{accents}

\title{Acquiring irrelevant information as a commitment\thanks{We would like to thank Daiki Kishishita for his helpful comments.}}
\author{Wataru Kitano\thanks{Graduate School of Management, Tokyo University of Science, 1-11-2 Fujimi, Chiyoda-ku, Tokyo 102-0071, Japan, wataru.kitano9998@gmail.com} and Shohei Yanagita\thanks{The Department of Economics, Takasaki City University of Economics, 1300 Kaminamie, Takasaki, Gumma 370-0801, Japan, shoheiyanagita@gmail.com}}
\date{\today}

\makeatother

\pgfplotsset{compat=1.18}
\begin{document}

\maketitle

\begin{abstract}
We formulate the voter's strategic information acquisition to control the future self's action as a Bayesian persuasion problem.
Our main result shows that acquiring information that is irrelevant to the voter's objective can be a \textit{worst-case optimal} solution: it can reduce the possibility that the future self is swayed by additional information whose content is ambiguous to the current voter.
\end{abstract}



\newpage

\section{Introduction}

Information acquisition is crucial for making better decisions. Since information provides guidance about the optimal choice, acquiring more information is generally considered desirable. However, individuals may sometimes be exposed to irrelevant information, which can confuse decision-making and lead them away from the optimal choice. In light of this concern, individuals may need to control the information they receive.

A natural example arises in elections. Voters aim to choose better politicians by gathering information about candidates before voting. The relevant information reflects candidates' political competence, such as their policy proposals, and voters are expected to base their decisions on this. However, in reality, decision-making can be influenced by irrelevant aspects of candidates, such as scandals, that are unrelated to their competence \citep[see, for example, ][]{rienks2023corruption}.

When such a temptation is foreseen, voters might try to control the information their future selves will access. For example, they may mute certain topics or block media sources that focus on scandals. Still, the future self might access additional irrelevant information beyond what was initially provided. Considering this, a natural conjecture is that the current voter should select an information structure that completely suppresses irrelevant information.

We show that this intuition need not hold. When there is uncertainty about whether the future self will acquire additional irrelevant information, the current voter may have an incentive to disclose it ex ante. If irrelevant information is entirely suppressed, the future self could conduct additional searches and come across information that further harms decision-making. Conversely, by collecting irrelevant information in a controlled way, the current voter can limit the room for gathering extra information.

We model this \textit{intrapersonal} game within an information design framework and establish a strategic foundation for demanding irrelevant information. More precisely, we adopt the framework of robust Bayesian persuasion proposed by \cite{dworczak2022preparing}. This framework is relatively recent and, to the best of our knowledge, has not yet been applied to self-control problems of this kind.
 
Our results complement the existing literature that provides alternative justifications for the acquisition of irrelevant information. We show that what is often viewed as a bias can, in fact, be a rational response to anticipated future behavior. This is a novel explanation for the acquisition of irrelevant information, which has been attributed to a receiver's bias or incentive to coordinate with other players.\footnote{See \cite{gentzkow2015media} for details about the receiver's bias and \cite{myatt2012endogenous} for information acquisition with coordination motives.}

\section{Model}
We consider a voter's self-control problem inspired by Example 1 (a) of \cite{jakobsen2025temptation}.
The current voter (henceforth CV) anticipates that the future self (henceforth FV) faces the problem of choosing whether to approve a candidate.
CV desires the approval decision to be based only on the content of the candidate's policy proposal.
However, CV is aware that FV may be affected by irrelevant features such as scandals.

To mitigate such a temptation, CV controls the information revealed to FV, for example, by customizing the news feed on CV's own social media account.
Moreover, CV anticipates that FV receives additional information generated through an unpredictable process, such as news recommended by a black-box algorithm on social media.

The model is a special case of \cite{dworczak2022preparing} with a particular player set, state space, action space, and payoffs.
There is a state space $\Omega=\Omega_R\times\Omega_I$ where $\Omega_R=\{G,B\}$ and $\Omega_I=\{g,b\}$, and each state is denoted by $\omega=(\omega_R,\omega_I)\in\Omega$.
CV and FV share the common prior $\mu_0\in\Delta\Omega$ such that $\mu_0(\omega)=1/4,\forall\omega\in\Omega$.
FV chooses an action $a$ from $A=\{1,0\}$, where action $1$ (resp. $0$) represents approval (resp. rejection) of a candidate.
CV's payoff function $v$ is defined on $A\times\Omega$ where $v(1,(G,g))=v(1,(G,b))=x$, $v(1,(B,g))=v(1,(B,b))=-y$, $v(0,\omega)=0,\forall\omega\in\Omega$, and $x,y>0$.
FV's payoff function $u$ is defined on $A\times\Omega$ where $u(1,(G,g))=x+m$, $u(1,(G,b))=x-n$, $u(1,(B,g))=m-y$, $u(1,(B,b))=-y-n$, $u(0,\omega)=0,\omega\in\Omega$, and $m,n>0$.
That is, CV desires the approval decision to be based only on whether $\omega_R$ is $G$ or $B$, while FV may be swayed by whether $\omega_I$ is $g$ or $b$.
Thus, the component $\Omega_R$ of the state represents whether the content of the candidate's policy proposal is good or bad for the voter, while $\Omega_I$ represents the feature that is irrelevant to the quality of the policy but affects FV's choice.

Define the set $\Pi$ of information structures.
That is, $\Pi$ consists of every tuple $(M,\pi)$ such that $M$ is a finite set of messages and $\pi:\Omega\to\Delta M$ is a function that maps each state realization $\omega$ to a distribution $\pi_\omega(\cdot)$ over messages in $M$.
Moreover, for each finite set $M$, define the set $\Pi'_M$ of \textit{additional} information structures.
That is, $\Pi'_M$ consists of each tuple $(M',\pi')$ such that $M'$ is a finite set of messages and $\pi':\Omega\times M\to\Delta M'$ maps each state $\omega$ and message $m\in M$ to the distribution $\pi'_{\omega,m}(\cdot)$ of the messages in $M'$.\footnote{
In \cite{dworczak2022preparing}, the message spaces $M$ and $M'$ are exogenously given, but assumed to be finite and sufficiently large.
Adopting this supposition, instead of our definition, does not affect the analysis.
}

In this game, first, CV chooses an information structure $(M,\pi)\in\Pi$. 
After that, Nature chooses the additional information structure $(M',\pi')\in\Pi'_M$.
Note that the distribution $\pi'$ can be contingent not only on the state realization $\omega\in\Omega$ but also on the message realization $m\in M$: this captures the possibility that FV may receive additional information that refines the original information $(M,\pi)$.

After CV and Nature choose information structures, a state $\omega\in\Omega$ is realized.
Then, a message $m\sim\pi_\omega(\cdot)$ is realized and FV observes $(M,\pi)$ and $m$, leading to Bayesian updating from the prior $\mu_0$ to the \textit{interim belief} $\mu_0^m$:
\[
\mu_0^m(\omega)=\frac{\pi_\omega(m)\mu_0(\omega)}{\sum_{\omega'\in\Omega}\pi_{\omega'}(m)\mu_0(\omega')}.
\]

Next, the message $m'\sim\pi'_{\omega,m}(\cdot)$ is realized and FV observes $(M',\pi')$ and $m'$, leading to Bayesian updating from the interim belief $\mu_0^m$ to the posterior $\mu_0^{m,m'}$:
\[
\mu_0^{m,m'}(\omega)=\frac{\pi'_{\omega,m}(m')\mu_0^m(\omega)}{\sum_{\omega'\in\Omega}\pi'_{\omega',m}(m')\mu_0^m(\omega')}.
\]
Finally, FV chooses an action $a\in A$.

For each belief $\mu$, define the set $A^*(\mu)\subseteq A$ of actions by 
\[
A^*(\mu):=\arg\max_{a\in A}\sum_{\omega\in\Omega}u(a,\omega)\mu(\omega).
\]
That is, $A^*(\mu)$ is the set of optimal actions for FV under the belief $\mu$.

\section{Analysis}
\subsection{Indirect payoffs}
First, CV's \textit{optimistic payoff} from each information structure $(M,\pi)\in\Pi$ is given by 
\[
\hat{v}(M,\pi):=\sum_{\substack{\omega\in\Omega\\m\in M}}\hat{V}(\mu_0^m)\pi_\omega(m)\mu_0(\omega),
\]
where 
\[
\hat{V}(\mu):=\max_{a\in A^*(\mu)}\sum_{\omega\in\Omega}v(a,\omega)\mu(\omega),\forall\mu\in\Delta\Omega.
\]
That is, $\hat{v}(M,\pi)$ is CV's indirect payoff from $(M,\pi)\in\Pi$ when Nature provides no information, i.e., $|M'|=1$, and FV's tie-breaking rule is optimal for CV.

Next, CV's \textit{pessimistic payoff} from each information structure $(M,\pi)\in\Pi$ is given by 
\[
\underaccent{\check}{v}(M,\pi):=\inf_{(M',\pi')\in\Pi'_M}\left[\sum_{\substack{\omega\in\Omega\\m\in M}}\left(\sum_{m'\in M'}\underline{V}(\mu_0^{m,m'})\pi'_{\omega,m}(m')\right)\pi_\omega(m)\mu_0(\omega)\right],
\]
where 
\[
\underline{V}(\mu):=\min_{a\in A^*(\mu)}\sum_{\omega\in\Omega}v(a,\omega)\mu(\omega).
\]
That is, $\underaccent{\check}{v}(M,\pi)$ is CV's indirect payoff from $(M,\pi)\in\Pi$ when Nature provides the worst additional information structure for CV and FV's tie-breaking rule is the worst for CV.

\subsection{Simplified problem}
Define the set $W\subseteq\Pi$ of information structures by
\[
W:=\arg\max_{(\bar{M},\bar{\pi})\in\Pi}\underaccent{\check}{v}(\bar{M},\bar{\pi}).
\]
\begin{definition}
An information structure $(M,\pi)\in\Pi$ is worst-case optimal if $(M,\pi)\in W$.
\end{definition}

\cite{dworczak2022preparing} provide a useful characterization of the set of the worst-case optimal information structures as follows.
Note that each message $m\in M$ that occurs with positive probability uniquely induces an interim belief $\mu^m$, and hence each information structure $(M,\pi)$ uniquely induces a distribution over interim beliefs.
Since $M$ is finite, each interim belief distribution induced by an information structure has finite support. 
Accordingly, from now on we focus on the set $\Delta_f \Delta\Omega$ of interim belief distributions with finite support, formally defined as
\[
\Delta_f \Delta\Omega=\left\{ \tau\in\Delta\Delta\Omega:|\mathrm{Supp}(\tau)|<\infty \right\}.
\]
Proposition 1 of \cite{kamenica2011bayesian} straightforwardly implies that a distribution $\tau\in\Delta_f \Delta\Omega$ is induced by some information structure in $\Pi$ if and only if $\tau$ satisfies \textit{Bayes plausibility}:
\begin{equation}\label{Eq_BP}\tag{BP}
\sum_{\mu\in\mathrm{Supp}(\tau)}\mu(\omega)\tau(\mu)=\mu_0(\omega),\forall\omega\in\Omega.
\end{equation}
Based on this observation, \cite{dworczak2022preparing} directly characterize the set of interim belief distributions, instead of the set of information structures.

To this end, CV's pessimistic payoff from each interim belief distribution $\tau\in\Delta_f\Delta\Omega$ induced by an information structure is given by 
\[
\sum_{\mu\in\mathrm{Supp}(\tau)}\underaccent{\check}{V}(\mu)\tau(\mu),
\]
where 
\[
\underaccent{\check}{V}(\mu):=\inf_{(M,\pi)\in\Pi}\sum_{\substack{\omega\in\Omega\\m\in M}}\underline{V}(\mu^m)\pi_\omega(m)\mu(\omega)
\]
and
\[
\mu^m(\omega)=\frac{\pi_\omega(m)\mu(\omega)}{\sum_{\omega'\in\Omega}\pi_{\omega'}(m)\mu(\omega')},\forall\omega\in\Omega.
\]
That is, $\underline{V}(\mu)$ is CV's pessimistic indirect payoff from each interim belief $\mu$ induced by an information structure.

As a benchmark, we define CV's indirect payoff from each interim belief $\mu$ when Nature fully reveals the state. 
Formally, it is given by 
\[
\underline{V}_{\text{full}}(\mu):=\sum_{\omega\in\Omega}\underline{V}(\delta_\omega)\mu(\omega),
\]
where $\delta_\omega\in\Delta\Omega$ is the belief that allocates probability one to the state $\omega$.
\cite{dworczak2022preparing} show that the worst-case optimal payoff coincides with the full-disclosure payoff.
\begin{lemma}{\citep{dworczak2022preparing}}
For each information structure $(M,\pi)\in\Pi$, the following are equivalent:
\begin{itemize}
\item $(M,\pi)\in W$.
\item $\tau_{(M,\pi)}\in\mathcal{W}$, where $\tau_{(M,\pi)}\in\Delta_f\Delta\Omega$ is an interim belief distribution induced by $(M,\pi)$ and
\[
\mathcal{W}:=\left\{\tau'\in\Delta_f\Delta\Omega\mid \sum_{\mu\in\mathrm{Supp}(\tau')}\underaccent{\check}{V}(\mu)\tau'(\mu)=\underline{V}_{\text{full}}(\mu_0) \text{ and } \tau' \text{ satisfies (\ref{Eq_BP})} \right\}.
\]
\end{itemize}
\end{lemma}
Although the proof follows straightforwardly from Lemma 1 of \cite{dworczak2022preparing}, we provide a proof because our statement slightly differs from theirs.
\begin{proof}
Pick an information structure $(M,\pi)\in\Pi$.
Note that
\[
\sum_{\mu\in\mathrm{Supp}(\tau_{(M,\pi)})}\underaccent{\check}{V}(\mu)\tau_{(M,\pi)}(\mu)\leq \underline{V}_{\text{full}}(\mu_0)
\]
because Nature can always fully reveal the state.
Thus, if $\tau_{(M,\pi)}\in\mathcal{W}$, then
\[
\sum_{\mu\in\mathrm{Supp}(\tau_{(M,\pi)})}\underaccent{\check}{V}(\mu)\tau_{(M,\pi)}(\mu)=\underline{V}_{\text{full}}(\mu_0),
\]
and hence $(M,\pi)\in W$.
If $\tau_{(M,\pi)}\not\in\mathcal{W}$, then $(M,\pi)\not\in W$ because
\begin{align*}
\underaccent{\check}{v}(M,\pi)&=\sum_{\mu\in\mathrm{Supp}(\tau_{(M,\pi)})}\underaccent{\check}{V}(\mu)\tau_{(M,\pi)}(\mu)\\ 
&<\underline{V}_{\text{full}}(\mu_0)\\
&=\underaccent{\check}{v}(\Omega,\pi^{\text{full}}),
\end{align*}
where $(\Omega,\pi^{\text{full}})\in\Pi$ is an information structure such that $\pi^{\text{full}}_{\omega}(\omega)=1,\forall\omega\in\Omega$, which fully reveals the true state.
\end{proof}

Next, they characterize the set $\mathcal{W}$ as follows.
For each function $v:\Delta\Omega\to\mathbb{R}$ and each set $Y\subseteq\Delta\Omega$, let $v\mid_Y:Y\to\mathbb{R}$ be a function such that $v\mid_Y(y)=v(y),\forall y\in Y$.
The following result offers a useful tool to identify what states are to be separated or pooled in worst-case optimal information structures:
\begin{lemma}{\citep{dworczak2022preparing}}\label{Lem_test}
It holds that 
\[
\mathcal{W}=\left\{\tau\in\Delta_f\Delta\Omega\mid \tau\text{ satisfies \eqref{Eq_BP}} \text{ and } \forall\mu\in\mathrm{Supp}(\tau),\mathrm{Supp}(\mu)\in\mathcal{F} \right\}
\]
where 
\[
\mathcal{F}=\left\{S\subseteq\Omega: \underline{V}\mid_{\Delta S}\geq\underline{V}_{\text{full}}\mid_{\Delta S} \right\}.
\]
\end{lemma}

\subsection{Main results}
We characterize the worst-case optimal solutions using Lemma \ref{Lem_test}.
The proofs are relegated to the Appendix.
\begin{proposition}\label{Prop_Gb}
$\{(G,g),(B,b)\},\Omega\setminus\{(B,g)\},\Omega\setminus\{(G,b)\},\Omega\not\in\mathcal{F}$.
\end{proposition}
This proposition states that, regardless of the parameters, information structures that pool $(G,g)$ and $(B,b)$ are never worst-case optimal.
Intuitively, if $(G,g)$ and $(B,b)$ are pooled, Nature easily sways FV's decision by exaggerating the probability of $(G,g)$ or $(B,b)$, which strongly incentivizes FV to approve or reject, respectively.
Therefore, such pooling is never optimal when Nature adversarially sways FV's decision.

\begin{proposition}\label{Prop_main}
The following holds:
\begin{itemize}
\item $\{(G,g),(G,b)\}\in\mathcal{F}$ if and only if $x\geq n$.
\item $\{(G,g),(B,g)\}\in\mathcal{F}$ if and only if $y\leq m$.
\item $\{(G,b),(B,g)\}\in\mathcal{F}$ if and only if $x\leq n$ or $y\leq m$.
\item $\{(B,g),(B,b)\}\in\mathcal{F}$ if and only if $y\geq m$.
\item $\{(G,b),(B,b)\}\in\mathcal{F}$ if and only if $x\leq n$.
\item $\Omega\setminus\{(B,b)\}\in\mathcal{F}$ if and only if $x\geq n$ and $y\leq m$.
\item $\Omega\setminus\{(G,g)\}\in\mathcal{F}$ if and only if $x\leq n$ and $y\geq m$.
\end{itemize}
\end{proposition}
Since
\[
\{(G,g)\},\{(G,b)\},\{(B,g)\},\{(B,b)\}\in\mathcal{F}
\]
for any parameters,
Proposition \ref{Prop_Gb} and \ref{Prop_main} provide a complete characterization of $\mathcal{F}$,
which implies the following important result.
\begin{theorem}\label{Th_main}
The following holds:
\begin{itemize}
\item If $n>x$, then, for any $S\subseteq\Omega$ such that $(G,g),(G,b)\in S$, it holds that $S\not\in\mathcal{F}$.
\item If $m>y$, then, for any $S\subseteq\Omega$ such that $(B,g),(B,b)\in S$, it holds that $S\not\in\mathcal{F}$.
\end{itemize}
\end{theorem}

This theorem provides a key implication:
if $m$ or $n$ is large, i.e., FV is more concerned about the irrelevant part $\Omega_I$ of the state, the worst-case optimal information structure for CV must separate $g$ and $b$ conditional on $B$ or $G$, respectively.
This means that, when CV pessimistically evaluates the additional information, CV may need to acquire information about the irrelevant part $\Omega_I$.

\begin{table}[t]
    \centering
    \caption{Worst-case optimal separations with the minimum number of messages.
    For instance, if $x\geq n$ and $y\geq m$, an information structure that separates the state space into $\{(G,g),(G,b)\}$ and $\{(B,g),(B,b)\}$ is worst-case optimal solution.}
    \begin{tabular}{|c|c|c|} \hline
         &$y\geq m$&$y<m$  \\ \hline
         $x\geq n$&$\{(G,g),(G,b)\},~~\{(B,g),(B,b)\}$&$\{(B,b)\},~~\Omega\setminus\{(B,b)\}$\\ \hline
         $x<n$&$\{(G,g)\},~~\Omega\setminus\{(G,g)\}$&$\{(G,g),(B,g)\},~~\{(G,b),(B,b)\}$ \\ \hline
    \end{tabular}
    \label{Tab_main}
\end{table}

Precisely, first, pick an interim belief $\mu$ such that $(G,g),(G,b)\in\mathrm{Supp}(\mu)$.
In this case, if the relevant part $\omega_R$ is likely to be $G$, the voter should approve.
However, if $n$ is large, Nature can sway this decision by exaggerating the probability that the irrelevant part $\omega_I$ is $b$.
In this sense, when $n$ is large, FV is susceptible to adversarial information that triggers a type-II error: refusing a good candidate.
To avoid this, the worst-case optimal solutions separate the irrelevant part $g$ and $b$ conditional on $G$.
In the same way, if $m$ is large, FV is susceptible to information that triggers a type-I error: approving a bad candidate.
To avoid this, the worst-case optimal solutions separate $g$ and $b$ conditional on $B$.

\textbf{Discussion.}
Table \ref{Tab_main} provides examples of the worst-case optimal solutions.
Importantly, for any parameters, there exists a simple worst-case solution  which uses only two messages. 

This observation tells us how the ambiguity can be mitigated with minimal information acquisition. 
For instance, if $x\geq n$ and $y<m$, separating $\{(B,b)\}$ from the other states is worst-case optimal.
That is, if FV is susceptible only to $g$, it suffices to acquire only the evidence of $(B,b)$ such as the news of corruption.\footnote{This is inspired by \cite{rienks2023corruption}, who uses corruption as an example of events that damage both the competency evaluation and the public image of politicians, which correspond to the relevant and irrelevant parts, respectively, in this paper.}
In a more extreme case, if both $m$ and $n$ are large, acquiring only irrelevant information can be rationalized.

\subsection*{Declaration of competing interests}
The authors declare that they have no competing financial interests or personal relationships that could have appeared to influence the work reported in this paper.

\section*{Appendix}

\subsection*{Omitted Proofs}
Note first that FV's expected payoff from the approval under each belief $\mu\in\Delta\Omega$ is given by
\begin{align*}
T(\mu)&=\sum_{\omega\in\Omega}u(1,\omega)\mu(\omega)\\
&=\mu(G,g)(x+m)+\mu(G,b)(x-n)+\mu(B,g)(m-y)-\mu(B,b)(y+n)\\
&=\left[\mu(G,g)+\mu(G,b)\right](x+y)+\left[\mu(G,g)+\mu(B,g)\right](m+n)-y-n,
\end{align*}
where the last equation follows from 
\[
\mu(B,b)=1-\mu(G,g)-\mu(G,b)-\mu(B,g).
\]
Then, it follows that
\begin{equation}\label{Eq_V_full}
\underline{V}_{\text{full}}(\mu)=\mu(G,g)x+\mu(G,b)1_{x>n}x-\mu(B,g)1_{m\geq y}y    
\end{equation}
and 
\begin{equation}\label{Eq_V_underbar}
\underline{V}(\mu)=\begin{cases}
    [\mu(G,g)+\mu(G,b)]x-[\mu(B,g)+\mu(B,b)]y,&\text{if }T(\mu)>0,\\
    \min\left\{0,[\mu(G,g)+\mu(G,b)]x-[\mu(B,g)+\mu(B,b)]y\right\}, &\text{if } T(\mu)=0,\\
    0, &\text{if } T(\mu)<0.
\end{cases}
\end{equation}

We show the propositions by testing whether each subset of the state space $\Omega$ belongs to $\mathcal{F}$.
To this end, we compare $\underline{V}_{\text{full}}$ and $\underline{V}$ below.
Calculation results are summarized in Table \ref{Tab_summary}.

\begin{table}[t]
    \centering
    \caption{Calculations of $\underline{V}_{\text{full}}(\mu)$ and $T(\mu)$ for belief $\mu$ supported by each nonsingleton subset of the state space $\Omega$.}
    \begin{tabular}{c|c|c}
         $\mathrm{Supp}(\mu)$&$\underline{V}_{\text{full}}(\mu)$&$T(\mu)$  \\\hline
         $\{(G,g),(G,b)\}$&$\substack{\mu(G,g)x\\+\mu(G,b)1_{x>n}x}$&$\substack{x-n\\+\mu(G,g)(m+n)}$\\ \hline
         $\{(G,g),(B,g)\}$&$\substack{\mu(G,g)x\\-\mu(B,g)1_{m\geq y}y}$&$\substack{m-y\\+\mu(G,g)(x+y)}$\\ \hline
         $\{(G,g),(B,b)\}$&$\mu(G,g)x$&$\substack{\mu(G,g)(x+y+m+n)\\-y-n}$\\ \hline
        $\{(G,b),(B,g)\}$&$\substack{\mu(G,b)1_{x>n}x\\-\mu(B,g)1_{m\geq y}y}$&$\substack{\mu(G,b)(x+y)\\+\mu(B,g)(m+n)-y-n}$\\ \hline
        $\{(G,b),(B,b)\}$&$\mu(G,b)1_{x>n}x$&$\substack{\mu(G,b)(x+y)\\-y-n}$\\ \hline
        $\{(B,g),(B,b)\}$&$-\mu(B,g)1_{m\geq y}y$&$\substack{\mu(B,g)(m+n)\\-y-n}$\\ \hline
        $\Omega\setminus\{(B,b)\}$&$\substack{\mu(G,g)x+\mu(G,b)1_{x>n}x\\-\mu(B,g)1_{m\geq y}y}$&$\substack{[\mu(G,g)+\mu(G,b)](x+y)\\+[\mu(G,g)+\mu(B,g)](m+n)-y-n}$\\ \hline
        $\Omega\setminus\{(B,g)\}$&$\substack{\mu(G,g)x\\+\mu(G,b)1_{x>n}x}$&$\substack{[\mu(G,g)+\mu(G,b)](x+y)\\+\mu(G,g)(m+n)-y-n}$\\ \hline
        $\Omega\setminus\{(G,b)\}$&$\substack{\mu(G,g)x\\-\mu(B,g)1_{m\geq y}y}$&$\substack{\mu(G,g)(x+y)\\+[\mu(G,g)+\mu(B,g)](m+n)-y-n}$\\ \hline
        $\Omega\setminus\{(G,g)\}$&$\substack{\mu(G,b)1_{x>n}x\\-\mu(B,g)1_{m\geq y}y}$&$\substack{\mu(G,b)(x+y)\\+\mu(B,g)(m+n)-y-n}$\\ \hline
        $\Omega$&$\substack{\mu(G,g)x+\mu(G,b)1_{x>n}x\\-\mu(B,g)1_{m\geq y}y}$&$\substack{[\mu(G,g)+\mu(G,b)](x+y)\\+[\mu(G,g)+\mu(B,g)](m+n)-y-n}$\\ \hline
    \end{tabular}
    \label{Tab_summary}
\end{table}

\subsubsection*{Proof of Proposition \ref{Prop_Gb}}
First, to show that $\{(G,g),(B,b)\}\not\in\mathcal{F}$, it suffices to show that there exists a belief $\mu$ such that $\mathrm{Supp}(\mu)=\{(G,g),(B,b)\}$ and $\underline{V}_{\text{full}}(\mu)>\underline{V}(\mu)$.
Such a belief $\mu$ is found as follows.
Let $\mu(G,g)=\epsilon$ and $\mu(B,b)=1-\epsilon$ for sufficiently small $\epsilon>0$.
Then, it follows that
\[
T(\mu)=\epsilon(x+y+m+n)-y-n<0
\]
and hence 
\[
\underline{V}_{\text{full}}(\mu)=\epsilon x>0=\underline{V}(\mu).
\]

Next, to show that $\Omega\setminus\{(B,g)\}\not\in\mathcal{F}$, let $\mu$ be a belief such that $\mathrm{Supp}(\mu)=\Omega\setminus\{(B,g)\}$ and $\mu(B,b)=1-\epsilon$ for sufficiently small $\epsilon>0$.
Then, it holds that $T(\mu)<0$ and hence 
\[
\underline{V}_{\text{full}}(\mu)=\mu(G,g)x+\mu(G,b)1_{x>n}>0=\underline{V}(\mu).
\]

Next, to show that $\Omega\setminus\{(G,b)\}\not\in\mathcal{F}$, let $\mu$ be a belief such that $\mathrm{Supp}(\mu)=\Omega\setminus\{(G,b)\}$, $\mu(G,g)=\epsilon$, and $\mu(B,g)=\epsilon^2$ for sufficiently small $\epsilon$. 
Then, it holds that $T(\mu)<0$ and hence 
\[
\underline{V}_{\text{full}}(\mu)=\epsilon x -\epsilon^2 1_{m\geq y}y>0=\underline{V}(\mu).
\]

Finally, to show that $\Omega\not\in\mathcal{F}$, let $\mu$ be a belief such that $\mathrm{Supp}(\mu)=\Omega$, $\mu(G,g)=\epsilon$, $\mu(G,b)=\epsilon^2$, $\mu(B,g)=\epsilon^3$, and $\mu(B,b)=1-\epsilon-\epsilon^2-\epsilon^3$ for sufficiently small $\epsilon>0$. 
Then, it holds that $T(\mu)<0$ and hence 
\begin{align*}
\underline{V}_{\text{full}}(\mu)&=\epsilon x +\epsilon^2 1_{x>n}x -\epsilon^3 1_{m\geq y}y \\
&>0\\
&=\underline{V}(\mu),
\end{align*}
which completes the proof.
\qed

\subsubsection*{Proof of Proposition \ref{Prop_main}}
First, suppose $\mathrm{Supp}(\mu)=\{(G,g),(G,b)\}$.
If $x\geq n$, it holds that $T(\mu)>0$ and hence
\[
\underline{V}(\mu)=x\geq \mu(G,g)x+\mu(G,b)1_{x>n}x=\underline{V}_{\text{full}}(\mu),
\]
which implies that $\{(G,g),(G,b)\}\in\mathcal{F}$.
If $x<n$, for some belief $\mu$ such that $\mu(G,g)=\epsilon$ for sufficiently small $\epsilon>0$, it holds that $T(\mu)<0$ and hence 
\[
\underline{V}_{\text{full}}(\mu)=\epsilon x>0=\underline{V}(\mu),
\]
which implies that $\{(G,g),(G,b)\}\not\in\mathcal{F}$.

Next, suppose $\mathrm{Supp}(\mu)=\{(G,g),(B,g)\}$.
If $m\geq y$, then $T(\mu)>0$, which implies that 
\[
\underline{V}(\mu)=\mu(G,g)x-\mu(B,g)y=\underline{V}_{\text{full}}(\mu)
\]
and hence $\{(G,g),(B,g)\}\in\mathcal{F}$.
If $m<y$, for a belief $\mu$ such that $\mu(G,g)=\epsilon$ for sufficiently small $\epsilon>0$, it holds that $T(\mu)<0$, which implies that 
\[
\underline{V}_{\text{full}}(\mu)=\mu(G,g)x>0=\underline{V}(\mu)
\]
and hence $\{(G,g),(B,g)\}\not\in\mathcal{F}$.

Next, suppose $\mathrm{Supp}(\mu)=\{(G,b),(B,g)\}$. 
If $x>n$ and $m\geq y$, then $T(\mu)>0$ and hence 
\[
\underline{V}(\mu)=\mu(G,b)x-\mu(B,g)y=\underline{V}_{\text{full}}(\mu),
\]
which implies that $\{(G,b),(B,g)\}\in\mathcal{F}$.
If $x\leq n$ and $m\geq y$, it holds that 
\[
\underline{V}(\mu)\geq -\mu(B,g)y=\underline{V}_{\text{full}}(\mu),
\]
which implies that $\{(G,b),(B,g)\}\in\mathcal{F}$.
If $x\leq n$ and $y>m$, it holds that $T(\mu)<0$ and hence 
\[
\underline{V}(\mu)=0=\underline{V}_{\text{full}}(\mu),
\]
which implies that $\{(G,b),(B,g)\}\in\mathcal{F}$.
If $x>n$ and $y>m$, for belief $\mu$ such that $\mu(G,b)=\epsilon$ for sufficiently small $\epsilon>0$, it holds that 
\begin{align*}
T(\mu)&=\epsilon(x+y) + (1-\epsilon)(m+n)-y-n\\
&=m-y+\epsilon(x+y-m-n)\\
&<0
\end{align*}
and hence
\[
\underline{V}_{\text{full}}(\mu)=\epsilon 1_{x>n}x>0=\underline{V}(\mu),
\]
which implies that $\{(G,b),(B,g)\}\not\in\mathcal{F}$.

Next, suppose $\mathrm{Supp}(\mu)=\{(G,b),(B,b)\}$.
If $x\leq n$, it holds that
\[
T(\mu)< x-n \leq 0
\]
and hence $\underline{V}(\mu)=0=\underline{V}_{\text{full}}(\mu)$, which implies that $\{(G,b),(B,b)\}\in\mathcal{F}$.
If $x>n$, for a belief $\mu$ such that $\mu(G,b)=\epsilon$ for sufficiently small $\epsilon>0$, it holds that $T(\mu)<0$ and hence 
\[
\underline{V}_{\text{full}}(\mu)=\mu(G,b)x>0=\underline{V}(\mu),
\]
which implies that $\{(G,b),(B,b)\}\not\in\mathcal{F}$.

Next, suppose $\mathrm{Supp}(\mu)=\{(B,g),(B,b)\}$.
If $y\geq m$, it holds that $T(\mu)< 0$ and hence 
\[
\underline{V}(\mu)=0\geq-\mu(B,g)1_{m\geq y}y=\underline{V}_{\text{full}}(\mu),
\]
which implies that $\{(B,g),(B,b)\}\in\mathcal{F}$.
If $m>y$, for belief $\mu$ such that $\mu(B,g)=1-\epsilon$ for sufficiently small $\epsilon>0$, it holds that $T(\mu)\geq 0$ and hence 
\[
\underline{V}_{\text{full}}(\mu)=-\mu(B,g)y>-y=\underline{V}(\mu),
\]
which implies that $\{(B,g),(B,b)\}\not\in\mathcal{F}$.

Next, suppose $\mathrm{Supp}(\mu)=\Omega\setminus\{(B,b)\}$.
If $x\geq n$ and $m\geq y$, it holds that
\begin{align*}
T(\mu)&=\mu(G,g)(x+y+m+n)+\mu(G,b)(x+y)+\mu(B,g)(m+n)-y-n \\
&>0
\end{align*}
and hence 
\begin{align*}
\underline{V}(\mu)&=[\mu(G,g)+\mu(G,b)]x-\mu(B,g)y\\
&\geq \underline{V}_{\text{full}}(\mu),
\end{align*}
which implies that $\Omega\setminus\{(B,b)\}\in\mathcal{F}$.
If $x<n$, for belief $\mu$ such that $\mu(G,g)=\epsilon$ and $\mu(B,g)=\epsilon^2$ for sufficiently small $\epsilon>0$, it holds that 
\begin{align*}
T(\mu)&=\epsilon(x+y+m+n)+(1-\epsilon-\epsilon^2)(x+y)+\epsilon^2(m+n)-y-n\\
&=x-n+\epsilon[m+n-\epsilon(x+y-m-n)]\\
&<0
\end{align*}
and hence 
\begin{align*}
\underline{V}_{\text{full}}(\mu)&=\epsilon x-\epsilon^2 1_{m\geq y}y\\
&>0\\
&=\underline{V}(\mu),
\end{align*}
which implies that $\Omega\setminus\{(B,b)\}\not\in\mathcal{F}$.
We can show that $y>m$ implies $\Omega\setminus \{(B,b)\}\not\in\mathcal{F}$ in the same way.

Finally, suppose $\mathrm{Supp}(\mu)=\Omega\setminus\{(G,g)\}$.
If $x\leq n$ and $y\geq m$, it holds that $T(\mu)<0$ and hence 
\[
\underline{V}(\mu)=0\geq
-\mu(B,g)1_{m\geq y}y=\underline{V}_{\text{full}}(\mu),
\] 
which implies that $\Omega\setminus\{(G,g)\}\in\mathcal{F}$.
If $x>n$ and $y\geq m$, for belief $\mu$ such that $\mu(G,b)=1-\epsilon$ for sufficiently small $\epsilon>0$, it holds that $T(\mu)>0$ and hence 
\begin{align*}
\underline{V}_{\text{full}}(\mu)&=\mu(G,b)x-\mu(B,g)1_{m\geq y}y\\
&>\mu(G,b)x-[\mu(B,g)+\mu(B,b)]y\\
&=\underline{V}(\mu),
\end{align*}
which implies that $\Omega\setminus\{(G,g)\}\not\in\mathcal{F}$.
If $x\leq n$ and $y< m$, for belief $\mu$ such that $\mu(G,b)=\epsilon^2$, $\mu(B,g)=1-\epsilon-\epsilon^2$, and $\mu(B,b)=\epsilon$ for sufficiently small $\epsilon>0$, it holds that $T(\mu)>0$ and hence 
\[
\underline{V}(\mu)=\epsilon^2 x- (1-\epsilon^2)y.
\]
Since
\[
\underline{V}_{\text{full}}(\mu)=-(1-\epsilon-\epsilon^2)y,
\]
it follows that 
\[
\underline{V}_{\text{full}}(\mu)-\underline{V}(\mu)=\epsilon y-\epsilon^2 x >0,
\]
which implies that $\Omega\setminus\{(G,g)\}\not\in\mathcal{F}$.
If $x>n$ and $y< m$, for belief $\mu$ such that $\mu(G,b)=1-\epsilon$ for sufficiently small $\epsilon$, it holds that $T(\mu)>0$ and hence 
\begin{align*}
\underline{V}_{\text{full}}(\mu)&=\mu(G,b)x-\mu(B,g)y\\
&>\mu(G,b)x-[\mu(B,g)+\mu(B,b)]y,
\end{align*}
which implies that $\Omega\setminus\{(G,g)\}\not\in\mathcal{F}$.\qed

\subsection*{Declaration of generative AI and AI-assisted technologies in the manuscript preparation process}
During the preparation of this work the authors used ChatGPT in order to improve the language and exposition of the manuscript. After using this tool, the authors reviewed and edited the content as needed and take full responsibility for the content of the published article.

\bibliographystyle{econ}
\bibliography{Literature}

\end{document}